\documentclass[letterpaper, 10 pt, conference]{ieeeconf}  % Comment this line out if you need a4paper

\IEEEoverridecommandlockouts                              % This command is only needed if 
\usepackage[usenames,dvipsnames]{xcolor}
\usepackage{psfrag,amsbsy,graphics,float}
\usepackage[dvips]{graphicx}
\usepackage{verbatim}

\usepackage{amsthm}
\usepackage{amsmath} % assumes amsmath package installed
\usepackage{amssymb}  % assumes amsmath package installed
\usepackage{tikz}
\usetikzlibrary{arrows,positioning,shapes}

\usepackage{pgfplots}
\usepackage{pgfplotstable}
\usepackage{booktabs}
\pgfplotsset{compat = newest}
\usepackage{hyperref}

\newtheorem{mytheo}{Theorem}
\newtheorem{myprop}{Proposition}
\newcommand{\R}{\mathbb{R}}
\newcommand{\N}{\mathbb{N}}
\newcommand{\C}{\mathcal{C}}
\newcommand{\X}{\mathcal{X}}

\newcommand{\bm}[1]{{\boldsymbol{#1}}}
\DeclareMathOperator{\var}{var}
\DeclareMathOperator{\rank}{rank}
\DeclareMathOperator{\kernel}{ker}

\newcommand{\GP}{\mathcal{GP}}

\title{\LARGE \bf
Stability of Gaussian Process State Space Models
}

\author{Thomas Beckers and Sandra Hirche% <-this % stops a space
\thanks{T. Beckers and S. Hirche are with the Chair of Information-oriented Control (ITR), Department of Electrical and Computer Engineering,
Technical University of Munich, D-80333 Munich\newline
{\tt\small \{t.beckers, hirche\}@tum.de}}% <-this % stops a space
}

\begin{document}

\maketitle
\thispagestyle{empty}
\pagestyle{empty}

%%%%%%%%%%%%%%%%%%%%%%%%%%%%%%%%%%%%%%%%%%%%%%%%%%%%%%%%%%%%%%%%%%%%%%%%%%%%%%%%
\begin{abstract}

Gaussian Process State Space Models (GP-SSMs) are a non-parametric model class suitable to represent nonlinear dynamics. They become increasingly popular in data-driven modeling approaches, i.e. when no first-order physics-based models are available. Although a GP-SSM produces well-behaved approximations and gains increasing popularity, the fundamental system dynamics are just sparsely researched. In this paper, we present stability results for the GP-SSM depending on selected covariance function  employing a deterministic point of view as widely done in the literature. The focus is set on the squared exponential function which is one of the most used covariance functions for nonlinear regression. We start with calculations according to the equilibrium points of GP-SSM and continue with conditions for stability.
\end{abstract}

%%%%%%%%%%%%%%%%%%%%%%%%%%%%%%%%%%%%%%%%%%%%%%%%%%%%%%%%%%%%%%%%%%%%%%%%%%%%%%%%
\section{Introduction}
Model-based control is a widespread technique for the control of dynamic systems. Most current  methods employ parametric model descriptions, i.e. for linear systems ARX or ARMAX. For nonlinear dynamics, the model-building process is much more complicated and it is often restricted to a specific class of systems. Just few popular approaches, such as NARMAX or Volterra series models, are available. A strong limitation of such identification strategies is that in many cases parametric models from first order physics laws are difficult to obtain. Especially, for complex systems such as human motion~\cite{wang2008gaussian} or gas-liquid separation~\cite{likar2007predictive} non-parametric techniques hold great promise. One popular possibility is to treat the systems as distribution over functions and use Gaussian Process State Space Models (GP-SSMs) to describe the nonlinear dynamic of the systems~\cite{frigola2014variational}. GP-SSMs become more and more popular in system identification for their favorable properties such as the bias variance trade-off and the strong Bayesian mathematics background.\\
A Gaussian Process generates data located throughout some domain such that any finite subset of the range follows a multivariate Gaussian distribution. This offers a powerful tool for nonlinear function regression with little prior knowledge~\cite{rasmussen2006gaussian}. The output of a GP prediction is a normal distributed variable which is uniquely defined by mean and variance. A Gaussian Process State Space Model is the application of a Gaussian Process to model a dynamic system, see e.g.~\cite{kocijan2005dynamic}. The GPs are trained by some input-output pairs of the system. Afterwards, they can estimate the mapping between the input and the output in untrained state space regions. Although Gaussian Process State Space Models become increasingly popular~\cite{kocijan2016gaussian} and start to be successfully used in control theory, e.g. for adaptive control~\cite{rogers2011adaptive}, the system properties of the GP-SSM are only sparsely researched.\\
In most of the works, where a GP-SSM is considered in a control setting, just the mean function of the process is employed, see e.g.~\cite{wang2005gaussian} and~\cite{chowdhary2013bayesian}. This is mainly because the GP is often used for replacing other deterministic methods. In order to provide rigorous guarantees on the system behavior, stability properties of GP-SSMs need to be well-understood, see e.g.~\cite{kocijan2003predictive} and~\cite{avzman2008non}. For linear system identification exists a stable kernel approach which include information on impulse response stability~\cite{chen2012estimation}. Chowdhary et al. presented a stability proof of an adaptive control approach with a Gaussian Process uncertainty model~\cite{chowdhary2012model} for nonlinear systems which is based on a bounded error model. A related model class are Gaussian Mixture Models (GMMs) which assume that every data point is generated from a mixture of a finite number of Gaussian distributions. Khansari-Zadeh et al. show for GMMs a syntheses approach for learning stable trajectories of a nonlinear dynamical system with GMMs~\cite{khansari2011learning}. In fact, it has been widely acknowledged, e.g. in~\cite{kocijan2005nonlinear}, that stability issues of GP-SSMs require careful attention in the future. The fundamental stability analysis of Gaussian Process State Space Models is still open.\\
The contribution of this paper is the study of equilibria of Gaussian Process State Space Models and their stability properties in terms of Lyapunov stability and ultimate boundedness. GP-SSMs with a linear, polynomial and the widespread squared exponential covariance function are analyzed. We determine the number of equilibrium points and present stability conditions for these models. The derived results are illustrated in numerical simulations.\\
The remainder of this paper is organized as follows: In Section II we introduce definitions about Gaussian Process State Space Models. In Section III the equilibrium points of GP-SSMs are analyzed. Stability conditions for GP-SSMs are presented in Section IV. Finally, Section V presents some illustrations of the previous proofs.

\textbf{Notation:} Bold characters are used for vectors and vector-valued functions. Matrices are denoted by capital letters. The expression~$\mathcal{N}(\mu,\Sigma)$ describes a normal distribution with mean~$\mu$ and covariance~$\Sigma$. The euclidean norm is given by~$\Vert\cdot\Vert$. The mean and variance of a probability variable is written as~$\mu(\cdot)$ and~$\var(\cdot)$.
%\subsection{Structure}

\section{Modeling with GP-SSMs}
In this section, we start with the necessary background information about Gaussian Processes and their application for GP-SSMs. 
%%%%%%%%%%%%%%%%%%%%%%%%%%%%%%%%%%%%%%%%%%%%%%%%%%
%%%%%%%%%%%%%%%%%%%%%%%%%%%%%%%%%%%%%%%%%%%%%%%%%%
\subsection{GP Definition}
Let~$(\Omega, \mathcal{F},P)$ be a probability space with the sample space~$\Omega$, the~$\sigma$-algebra~$\mathcal{F}$ and the probability measure~$P$. The set~$\X \subseteq \R^n$ with~$n\in\N^*$ denotes a corresponding index set. A stochastic process is a discrete or real valued function~$f(\bm x, \omega)$ that for every fixed~$\bm x\in\X$ is a measurable function of~$\omega\in\Omega$. For fixed~$\omega\in\Omega$, the function~$f(\bm x, \omega)$ becomes a deterministic function of~$\bm x$. This function is known as sample path or realization of the stochastic process. If~$\bm x\in\X$ is fixed, the function~$f(\bm x, \omega)$ is a random variable on~$\Omega$. A Gaussian Process is such a stochastic process which can also be interpreted as a distribution over functions. Therefore, it describes a probability distribution over an infinite dimensional vector space. Gaussian Processes are fully specified by a mean function~$m(\bm x)\in\C^0$ and a covariance function~$k(\bm x,\bm x^\prime)\in\C^0$, which is also known as kernel function. The elements of the index set~$\X$ are called states.
\begin{align}
&f(\bm x) \sim \GP(m(\bm x),k(\bm x,\bm x^\prime)),\qquad \bm x,\bm x^\prime\in\X\\
&m(\bm x)\colon\X\to\R,\,k(\bm x,\bm x^\prime)\colon\X\times \X\to\R
\end{align}
The value of the covariance function~$k(\bm x,\bm x^\prime)$ is an indicator of the interaction of two states~$(\bm x,\bm x^\prime)$.  
In practice, the mean function is often set to zero, as this simplifies calculations without limiting the expressive power of the process. The choice of the covariance function and its parameters is a degree of freedom of the GP regression. The essential part in GP model learning is the selection of the function~$k(\bm x,\bm x^\prime)$ and the estimation of its free parameters~$\varphi$, called hyperparameters. Common covariance functions include the squared exponential, the linear, and the polynomial covariance function, see Table~\ref{tab:kernel}.
\renewcommand{\arraystretch}{1.5}
\begin{table}[h]
\begin{center}
\begin{tabular}{|p{2.4cm}|l|p{2.2cm}|}
\hline 
Covariance function &$k(\bm x,\bm x^\prime)=$ &  hyperparameters $\varphi$\\ 
\hline 
linear & $\bm x^\top \bm x^\prime+ \sigma_0^2$ & $\{\sigma_0\in\R_+\}$ \\ 
\hline 
polynomial & $\left(\bm x^\top \bm x^\prime+ \sigma_0^2 \right) ^p$ & \parbox[c][0.9cm]{1.4cm}{\begin{align*} \{\sigma_0&\in\R_+,\\p&\in\N\vert p\geq 2\}\end{align*}}\\ 
\hline 
squared exponential & $\sigma_f^2 \exp{\left(-\frac{\Vert \bm x- \bm x^\prime \Vert^2}{2\lambda^2} \right) }$ & \parbox[c][0.9cm]{1.4cm}{\begin{align*} \{\sigma_f&\in\R_+,\\ \lambda&\in\R_+^*\} \end{align*}}\\
\hline 
\end{tabular} 
\end{center}
\caption{Summary of some commonly-used covariance functions.\label{tab:kernel}}
\end{table}
\renewcommand{\arraystretch}{1}

The hyperparameters~$\sigma_0^2$ and~$\sigma_f^2$ describe the signal variance which determines the average distance of the function~$f(\bm x)$ away from its mean. A Gaussian Process with a linear covariance function is a Bayesian linear regression with variance~$\sigma_0^2$. More flexibility provides the polynomial function which allows to learn nonlinear models. Probably the most widely used covariance function in machine learning is the squared exponential covariance function, see~\cite{rasmussen2006gaussian}, with the related hyperparameters~$\{\sigma_f,\lambda\}$. The length-scale~$\lambda$ determines the number of expected upcrossing of the level zero in a unit interval by a zero-mean GP. This covariance function is infinitely differentiable which means that the GP exhibits a smooth behavior. A more detailed discussion about the advantages of different kernel functions can be found, for example, in~\cite{mackay1997gaussian} and~\cite{bishop2006pattern}.
%%%%%%%%%%%%%%%%%%%%%%%%%%%%%%%%%%%%%%%%%%%%%%%%%%
%%%%%%%%%%%%%%%%%%%%%%%%%%%%%%%%%%%%%%%%%%%%%%%%%%
\subsection{Gaussian Process State Space Models}
A Gaussian Process State Space Model for autonomous, discrete systems maps the current state~$\bm x_k$ to the next step ahead state~$\bm x_{k+1}$.
\begin{align}
\begin{split}
\bm x_{k+1}&=\bm f(\bm x_k),\qquad k\in\N\\
\bm f(\bm x_k)&\sim \GP(\bm m(\bm x_k),\bm k(\bm x_k,\bm x^\prime_k))
\end{split}
\label{for:GPSSS}
\end{align}
where the vector~$\bm x_k\in\X$ represents the state of the system. The vector function~$\bm m(\cdot)=[m_1(\cdot),\ldots,m_n(\cdot)]^\top$ contains the mean functions for each component of~$\bm x_{k+1}$. The function~$\bm k(\cdot,\cdot)=[k_{\varphi_1}(\cdot,\cdot),\ldots,k_{\varphi_n}(\cdot,\cdot)]^\top$ is composed of covariance functions where~$\varphi_i$ is the corresponding set of hyperparameters, see Table~\ref{tab:kernel}.  Due to the fact, that the GP can only map to a one dimensional space, a~$n$-dimensional system needs~$n$ GPs. So the representation (\ref{for:GPSSS}) is defined by
\begin{align}
\bm f(\bm x_k)=
\begin{cases} 
f_1(\bm x_k)\sim \GP(m_1(\bm x_k),k_{\varphi_1}(\bm x_k,\bm x^\prime_k))\\
\vdots\hspace{0.9cm}\vdots\hspace{0.5cm}\vdots\\
f_n(\bm x_k)\sim \GP(m_n(\bm x_k),k_{\varphi_n}(\bm x_k,\bm x^\prime_k)).
\end{cases}
\end{align}
To predict~$\bm x_{k+1}$ for a given~$\bm x_k$ the GP-SSM is trained with training input and output pairs. Suppose, we set the mean~$\bm m(\cdot)=\bm 0$ and we have~$m$ training inputs~$\{\bm{\tilde{x}}_{j_i}\}_{i=1}^m$ and outputs~$\{\bm{\tilde{x}}_{j_i+1}\}_{i=1}^m$ pairs with~$j_i\in\N,\bm{\tilde{x}}\in\X$. We arrange the data in an input training matrix which is defined by~$X=[\bm{\tilde{x}}_{j_1},\bm{\tilde{x}}_{j_2},\ldots,\bm{\tilde{x}}_{j_m}]$ and an output training matrix~$Y^\top=[\bm{\tilde{x}}_{j_1+1},\bm{\tilde{x}}_{j_2+1},\ldots,\bm{\tilde{x}}_{j_m+1}]$. Using the marginalization property, the prediction for each component of the one step ahead state vector~$x_{i,k+1}$ is calculated as Gaussian distributed variable with the mean~$\mu(x_{i,k+1})$ and the variance~$\var(x_{i,k+1})$. The joint distribution of the~$i$-th component of the predicted next step ahead state~$x_{i,k+1}$ and the corresponding vector of the training outputs~$Y$ is 
\begin{align}
\begin{bmatrix} Y_{1\ldots m,i} \\ x_{i,k+1} \end{bmatrix}\sim \mathcal{N} \left(\bm 0, \begin{bmatrix} K_{\varphi_i}(X,X) & \bm k_{\varphi_i}(\bm x_k,X)\\ \bm k_{\varphi_i}(\bm x_k,X)^\top & k_{\varphi_i}(\bm x_k,\bm x_k) \end{bmatrix}\right)
\end{align} 
where~$Y_{1\ldots m,i}$ is the~$i$-th column of the matrix~$Y$. The function~$K_{\varphi_i}(X,X)$ is called covariance matrix, and~$\bm k_{\varphi_i}(\bm x_k,X)$ the vector-valued extended covariance function with the set of hyperparameters~$\varphi_i$. They are defined by
\begin{align}
\begin{split}
&K_{\varphi_l}(X,X)\colon\X^m\times \X^m\to\R^{m\times m}\\
&K_{i,j}= k_{\varphi_l}(X_{1\ldots n, i},X_{1\ldots n, j})\\
&\bm k_{\varphi_l}(\bm x_k,X)\colon\X\times \X^m\to\R^m,\,k_{i} = k_{\varphi_l}(\bm{x}_k,X_{1\ldots n, i})\\
&\forall i,j\in\lbrace 1,\ldots,m\rbrace,l\in\lbrace 1,\ldots,n\rbrace .
\end{split}
\end{align}
A prediction of the~$i$-th component of~$\bm x_{k+1}$ is produced with
\begin{align}
x_{i,k+1} &\sim \mathcal{N} \left(\mu_i(\bm x_{k+1}), \var_i(\bm x_{k+1})\right),\\
\mu_i(\bm x_{k+1}\vert\bm x_{k})&=\bm k_{\varphi_i}(\bm x_k,X)^\top (K_{\varphi_i}(X,X)+I \sigma^2_{n,i})^{-1}\notag\\
& \phantom{{}=}Y_{1\ldots m,i}\label{for:meanvalue}\\
\var_i(\bm x_{k+1}\vert\bm x_{k})&=k_{\varphi_i}(\bm x_k,\bm x_k)-\bm k_{\varphi_i}(\bm x_k,X)^\top \notag\\
& \phantom{{}=}K_{\varphi_i}^{-1}(X,X) \bm k_{\varphi_i}(\bm x_k,X). 
\end{align} 
where~$\mu_i(\cdot)$ is the mean and~$\var_i(\cdot)$ the variance of the random variable. The addition of~$\sigma_{n,i}^2\in\R_+^*,\forall i\in\{1,\ldots,m\}$ allows the algorithm to handle noisy input data. Besides, the numerical stability of the matrix inversion increases. The~$n$ normal distributed components are combined in a multi-variable distribution.
\begin{align}
\bm x_{k+1}&\sim \mathcal{N} \left( \bm \mu(\bm x_{k+1}), \var(\bm x_{k+1})I\right)\\
\bm \mu(\bm x_{k+1}\vert\bm x_{k})&=[\mu_i(\bm x_{k+1}),\ldots,\mu_n(\bm x_{k+1})]^\top\\
\bm \var(\bm x_{k+1}\vert\bm x_{k})&=[\var_i(\bm x_{k+1}),\ldots,\var_n(\bm x_{k+1})]^\top
\end{align} 
%%%%%%%%%%%%%%%%%%%%%%%%%%%%%%%%%%%%%%%%%%%%%%%%%%
%%%%%%%%%%%%%%%%%%%%%%%%%%%%%%%%%%%%%%%%%%%%%%%%%%
\section{Equilibrium points of GP-SSMs}
In this section, we analyze the GP-SSM in terms of the existence of equilibrium points. In the following, we focus on the deterministic point of view. Therefore, just the mean prediction~$\bar{\bm x}_{k+1}=\bm \mu(\bm x_{k+1})$ is taken into account (deterministic GP-SSM). We call the set of equilibrium points of a discrete-time system~$\X^*$ with
\begin{align}
\X^*=\left\lbrace \bm x^*\in\X\mid \bm x^*=\bm f(\bm x^*) \right\rbrace .
\end{align}
The cardinality~$\vert X^* \vert$ is the number of equilibrium points. Each component of the predicted state vector of a deterministic GP-SSM, see (\ref{for:meanvalue}), can be written as weighted sum of covariance functions. The number of covariance functions is equal to the number of training points~$m$.
\begin{align}
\bar x_{i,k+1}&\hspace{-0.1cm}=\hspace{-0.1cm}\sum_{j=1}^m{k_{j,\varphi_i}(\bm x_k,X)\underbrace{[(K_{\varphi_i}(X,X)+I\sigma^2_{n,i})^{-1}Y_{1\ldots m,i}]_j}_{h_j(i)}} \label{for:meanwithh}
\end{align} 
The vector of weighting factors~$\bm h(i)=[h_1(i),\ldots,h_m(i)]^\top$ depends on the inverse of the covariance matrix with signal noise~$(K_{\varphi_i}(X,X)+I\sigma^2_{n,i})^{-1}$, the output training matrix~$Y$ and the required component~$i$.\\
The following gives an overview about the behavior of the different covariance functions presented in Table~\ref{tab:kernel}.
%%%%%%%%%%%%%%%%%%%%%%%%%%%%%%%%%%
%Squared exponential
%%%%%%%%%%%%%%%%%%%%%%%%%%%%%%%%%%
\subsection{Squared exponential covariance function}
The often used squared exponential covariance function~$k(\bm x,\bm x^\prime)=\sigma_f^2 \exp{\left(-\Vert \bm x- \bm x^\prime \Vert^2/(2\lambda^2) \right) }$ is very powerful for nonlinear function regression. The following theorem gives a lower bound of the quantity of equilibrium points.
\begin{myprop}
The set of equilibrium points of deterministic GP-SSMs with squared exponential covariance function has at least one equilibrium point
\begin{align*}
\min\vert X^* \vert=1.
\end{align*}
\end{myprop} 
\begin{proof}
The idea of the proof is that each single equation~$x_{i,k}^*=f_i(\bm x_k)$ has a solution for any fixed component~$x_{j,k}$ with~$j\in\{1,\ldots,n\}$ and~$j\neq i$. Therefore, it must exist at least one solution for the overall system of equations.\\
For the proof of the minimum quantity of equilibrium points, we consider (\ref{for:meanwithh}) and insert the squared exponential covariance function
\begin{align}
x_{i,k+1}&=\sum_{j=1}^m{\sigma_{i,f}^2\exp\left(-\frac{\Vert \bm x_k-X_{1\ldots n,j} \Vert^2}{2\lambda_i^2}\right) h_j(i)}.\label{for:expfcn}
\end{align}
The parameters~$\sigma_{i,f}$ and~$\lambda_i$ are the corresponding hyperparameters of the function~$f_i(\cdot)$. As far as the authors know, it is not possible to find an analytic solution for this kind of multivariate equation system. Therefore, the system functions will be treated separately. This kind of view neglects the effects of the multivariate structure but provides also a valid solution. An important property of the squared exponential function is the behavior at the limit:
\begin{align}
\lim_{\Vert \bm x \Vert \to\infty}\sigma_f^2 \exp{\left(-\frac{\Vert \bm x-\bm x^\prime \Vert^2}{2\lambda^2} \right)}&=0,\qquad \text{with }\bm x^\prime\in\R^n
\end{align}
Since the limit of the squared exponential function is zero, the limit of the weighted sum of squared exponential functions must be also zero.
\begin{align}
\lim_{\Vert \bm x \Vert \to\infty}\sum_{j=1}^m{\sigma_{i,f}^2\exp\left(-\frac{\Vert \bm x-X_{1\ldots n,j} \Vert^2}{2\lambda_i^2}\right) h_j(i)}=0
\end{align}
We recall Bolzano's theorem which is a special case of the intermediate value theorem.
\begin{mytheo}[Bolzano, \cite{larson2013calculus}]
Suppose~$f(x):[a,b]\to \R$ is continuous on the closed interval~$[a,b]$ and suppose that~$f(a)$ and~$f(b)$ have opposite signs. Then there exists a number~$c$ in the interval~$[a,b]$ for which~$f(c)=0$.
\end{mytheo}

Since Bolzano's theorem just holds for scalar functions, (\ref{for:expfcn}) must be rewritten as function of a scalar variable. For this purpose, the components~$x_j^*$ with~$j\in\{1,\ldots,n\}\vert j\neq i$ are fixed. The resulting function is called~$\bm f^s(x_i)\colon \R\to\R^n$.
\begin{align}
f_i^s(x_{i,k})&:=f_i([x_{1,k},\ldots,x_{i-1,k},x_{i,k},x_{i+1,k},\ldots,{x_n,k}])\notag\\
&\text{  with fixed } x_{1,k},\ldots,x_{i-1,k},x_{i+1,k},\ldots,x_{n,k}\in\R 
\end{align}
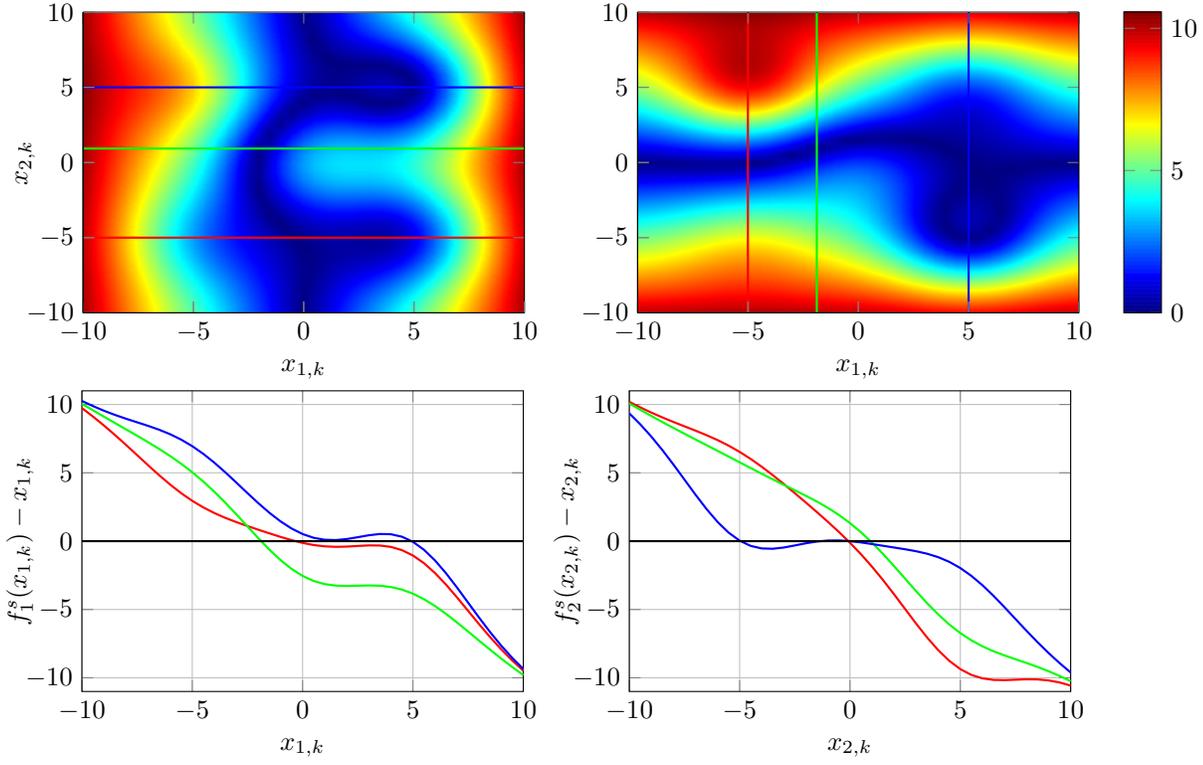
\begin{figure*}[thpb]
\begin{tikzpicture}
    \begin{axis}[width=0.33\textwidth,
    				height=4cm,
    				view={0}{90}, 
    				colormap/jet, 
    				shader=interp,
    				xlabel={$x_{1,k}$},
    				ylabel={$x_{2,k}$},
    				ylabel shift=-10pt,
    				scale only axis]
    \addplot3[surf,mesh/cols=41] table {data/data1.dat};
    \addplot[blue,thick] coordinates {(-10,5) (10,5)};
    \addplot[green,thick] coordinates {(-10,0.93) (10,0.93)};
    \addplot[red,thick] coordinates {(-10,-5) (10,-5)};
    \end{axis}
\end{tikzpicture} 
\hspace*{0.1cm}
\begin{tikzpicture}    
    \begin{axis}[width=0.33\textwidth,
    				height=4cm,
    				view={0}{90}, 
    				colormap/jet, 
    				shader=interp,
    				xlabel={$x_{1,k}$},
    				ylabel shift=-10pt,
    				colorbar,
    				scale only axis]
    \addplot3[surf,mesh/cols=41] table {data/data2.dat};
    \addplot[blue,thick] coordinates {(5,-10) (5,10)};
    \addplot[green,thick] coordinates {(-1.88,-10) (-1.88,10)};
    \addplot[red,thick] coordinates {(-5,-10) (-5,10)};
    \end{axis}
\end{tikzpicture}\\ 
\hspace*{-0.25cm}
\begin{tikzpicture}
    \begin{axis}[width=0.33\textwidth,
    				height=4cm,
    				xlabel={$x_{1,k}$},
    				ylabel={$f_1^s(x_{1,k})-x_{1,k}$},
    				xmin = -10,
    				xmax = 10,
    				ymin = -11,
    				ymax = 11,
    				grid = major,
    				ylabel shift=-10pt,
    				scale only axis]
    \addplot[red,thick] table[y index = 1] {data/data3.dat};
    \addplot[green,thick] table[y index = 2] {data/data3.dat};
    \addplot[blue,thick] table[y index = 3] {data/data3.dat};
    \addplot[black,thick] coordinates {(-10,0) (10,0)};
    \end{axis}
\end{tikzpicture}
\hspace*{-0.25cm}
\begin{tikzpicture}
    \begin{axis}[width=0.33\textwidth,
    				height=4cm,
    				xlabel={$x_{2,k}$},
    				ylabel={$f_2^s(x_{2,k})-x_{2,k}$},
    				xmin = -10,
    				xmax = 10,
    				ymin = -11,
    				ymax = 11,
    				grid = major,
    				ylabel shift=-10pt,
    				scale only axis]
    \addplot[red,thick] table[y index = 1] {data/data4.dat};
    \addplot[green,thick] table[y index = 2] {data/data4.dat};
    \addplot[blue,thick] table[y index = 3] {data/data4.dat};
    \addplot[black,thick] coordinates {(-10,0) (10,0)};
    \end{axis}
\end{tikzpicture} 
\caption{Top: Shows color-coded~$\vert f_1(\bm x_k)-x_{1,k}\vert$ on the left and~$\vert f_2(\bm x_k)-x_{2,k}\vert$ on the right against~$x_{1,k}$ and~$x_{2,k}$. Dark blue marks the area with possible equilibrium points. Bottom: On the left side,~$x_{2,k}$ is fixed by three example values~$-5$ (red),~$5$ (blue) and~$0.93$ (green). On the right side,~$x_{1,k}$ is fixed by three example values~$-5$ (red),~$5$ (blue) and~$-1.88$ (green). As Bolzano's theorem predicts, each function has at least one zero crossing. \label{fig:sdexp}}
\end{figure*}
Due to the fact that~$f_i^s(\cdot)$ is continuous and the limits
\begin{align}
\lim_{x_{i,k} \to\infty} f_i^s(x_{i,k})-x_{i,k}=-\infty\\
\lim_{x_{i,k} \to-\infty} f_i^s(x_{i,k})-x_{i,k}=\infty
\end{align}
have a different sign, Bolzano's theorem predicts at least one solution~$x_i^*$ for~$f_i^s(x_i^*)=x_i^*$. Since this holds for any~$x_{1,k},\ldots,x_{i-1,k},x_{i+1,k},\ldots,x_{n,k}\in\R$, each function~$f_i^s(x_i^*)$ has such a solution. Therefore, there must exist an equilibrium point~$\bm x^*$ which fulfils~$\bm x^*=\bm f(\bm x^*)$.
\end{proof}
Figure \ref{fig:sdexp} demonstrates the idea of the proof. For an example system with two states, the top row shows color-coded on the left side the difference between~$f_1(\bm x_k)$ and~$x_{1,k}$ and on the right side the difference between~$f_2(\bm x_k)$ and~$x_{2,k}$. If the distance is zero, which is illustrated by dark color, the component of the state vector~$x_{i,k}$ equals~$f_i(\bm x_k)$. The second row shows the slice plane~$f_i^s(x_{i,k})-x_{i,k}$ which should be zero for an equilibrium. On the left side,~$x_{2,k}$ is fixed by three example values~$-5$ (red),~$5$ (blue) and~$0.93$ (green). On the right side,~$x_{1,k}$ is fixed by three example values~$-5$ (red),~$5$ (blue) and~$-1.88$ (green). As Bolzano's theorem predicts, each function has at least one zero crossing. Therefore, it is possible to find two values~$x^*_1$ and~$x^*_2$ which fulfill~$f_i^s(x_i^*)-x_i^*=0$ for each~$i\in\{1,2\}$. For this example system a numerical solver determinates one equilibrium point at~$\bm x^*=[-1.88,0.93]^\top$. The green function illustrates this value. On the left side, the function crosses zero at~$0.93$ and on the right side zero is crossed at~$-1.88$.
%%%%%%%%%%%%%%%%%%%%%%%%%%%%%%%%%%
%linear
%%%%%%%%%%%%%%%%%%%%%%%%%%%%%%%%%%
\subsection{Linear covariance function}
The next analysis is about the equilibrium points of the linear covariance function~$k(\bm x,\bm x^\prime)=\bm x^\top \bm x^\prime+ \sigma_0^2$. 
\begin{myprop}
The set of equilibrium points of deterministic GP-SSMs with linear covariance function has the following properties:
\begin{align*}
\vert X^* \vert =0 \;\vee\; \vert X^* \vert =1 \;\vee\; \vert X^* \vert =\infty
\end{align*}
\end{myprop}
\begin{proof}
We start with (\ref{for:meanwithh}) and use the linear covariance function.
\begin{align}
\bar x_{i,k+1}&=\sum_{j=1}^m{k_{j,\varphi_i}(\bm x_k,X)h_j(i)}\\
&=\sum_{j=1}^m{(\bm x_k^\top X_{1\ldots n,j} +\sigma_{i,0}^2)h_j(i)}\\
&=\sum_{j=1}^m{\bm x_k^\top X_{1\ldots n,j} h_j(i)}+\sigma_{i,0}^2 \bm 1^\top \bm h(i)
\end{align}
Since the sum of linear functions is also a linear function, the whole one step ahead state vector~$\bar{\bm x}_{k+1}$ is denoted by
\begin{align}
\bar{\bm x}_{k+1}&=\underbrace{\begin{pmatrix}
X_{1,1\dots m}\bm h(1),\ldots,X_{n,1\dots m}\bm h(1)\\
X_{1,1\dots m}\bm h(2),\ldots,X_{n,1\dots m}\bm h(2)\\
\vdots\\
X_{1,1\dots m}\bm h(n),\ldots,X_{n,1\dots m}\bm h(n)\\
\end{pmatrix}}_{A}\bm x_k \notag\\
&+
\underbrace{\begin{pmatrix}
\sigma_{1,0}^2 \bm 1^\top \bm h(1)\\
\sigma_{2,0}^2 \bm 1^\top \bm h(2)\\
\vdots\\
\sigma_{n,0}^2 \bm 1^\top \bm h(n)
\end{pmatrix}}_{\bm b}
\label{for:stab:lin}
\end{align}
and is written as non-homogeneous linear system with state matrix~$A\in \R^{n\times n}$ and offset~$b\in \R^{n}$. The equilibrium points are calculated by solving the equation~$\bm x^*=A \bm x^*+\bm b$ with~$A^I=I-A$ and the Moore-Penrose pseudoinverse matrix~$(A^I)^+$ the set of equations may behave in any one of three possible ways:
\begin{itemize}
\item[(i)] The system has a single unique solution if~$\rank(A^I)=\rank(A^I|\bm b)=n$ $\Rightarrow$ $\X^*=\{(A^I)^{-1}\bm b\}$
\item[(ii)] The system has infinitely many solutions if~$\rank(A^I)=\rank(A^I|\bm b)<n$ $\Rightarrow$ $\X^*=\{(A^I)^{+}\bm b+\kernel(A)\}$
\item[(iii)] The system has no solution if~$\rank(A^I)\neq\rank(A^I|\bm b)$ $\Rightarrow$ $\X^*=\{\varnothing\}$
\end{itemize}
\end{proof}
Due to the fact that the presented conditions (ii) and (iii) are very unlikely, a system with infinitely many solutions or no solution is in practice as good as impossible. For example, if we assume a one dimensional system, $A$ must be exactly~$1$ to obtain infinitely many solutions (if $b=0$) or no solution (if $b\neq 0$).
%%%%%%%%%%%%%%%%%%%%%%%%%%%%%%%%%%
%polynomial
%%%%%%%%%%%%%%%%%%%%%%%%%%%%%%%%%%
\subsection{Polynomial covariance function}
The second studied function is the polynomial covariance function~$k(\bm x,\bm x^\prime)=\left(\bm x^\top \bm x^\prime+ \sigma_0^2 \right) ^p$ which is more flexible and allows nonlinear function estimation. The degree~$p$ is important for the quantity of equilibrium points as the next theorem shows.
\begin{myprop}
The set of equilibrium points of deterministic GP-SSMs with polynomial covariance function has the following properties
\begin{align*}
\max\vert X^* \vert =\prod_{i=1}^n {p_i}
\end{align*}
where~$p_i$ is the degree of the corresponding covariance function to the~$i$-th component of~$\bar{ \bm x}_{k+1}$.
\end{myprop}
\begin{proof}
We use again (\ref{for:meanwithh}) and insert the polynomial covariance function
\begin{align}
\bar x_{i,k+1}&=\sum_{j=1}^m{k_{j,\varphi_i}(\bm x_k,X)h_j(i)}\\
&=\sum_{j=1}^m{(\bm x_k^\top X_{1\ldots n,j} +\sigma_{i,0}^2)^{p_i} h_j(i)}\label{for:polfcn}
\end{align}
where~$\bm p=[p_1,\ldots,p_n]^\top\in\N^n$ contains the degree of each covariance function. With the multinomial theorem and the condition for equilibrium points~$\bm x^*=\bm f(\bm x^*)$, equation (\ref{for:polfcn}) can be written as
\begin{align}
x^*_i&=\sum_{l_1+\ldots+l_{n+1}=p_i}{\alpha_{i,l_1,\ldots,l_{n+1}} x_1^{*^{l_1}} x_2^{*^{l_2}} \cdots x_n^{*^{l_n}}\sigma_{i,0}^{2l_{n+1}}}
\end{align}
with~$0\leq l_1,\ldots,l_{n+1} \leq n$ and~$\alpha_{l_1,\ldots,l_{n+1}}\in\R$. The term of the left-hand side can be integrate in the right-hand side by adapting the coefficients~$\alpha_{l_1,\ldots,l_{n+1}}$ to~$\beta_{i,l_1,\ldots,l_{n+1}}\in\R$.
 \begin{align}
0&=\sum_{l_1+\ldots+l_{n+1}=p_i}{\beta_{i,l_1,\ldots,l_{n+1}} x_1^{*^{l_1}} x_2^{*^{l_2}} \cdots x_n^{*^{l_n}}\sigma_{i,0}^{2l_{n+1}}}\label{for:poly}
\end{align}
The theorem of B\'ezout gives an upper bound for the number of roots for this polynomial system.
\begin{mytheo}[B\'ezout, \cite{sturmfels1998polynomial}]
Unless a square polynomial system denoted by~$\bm f(\bm x)$ with degree~$d_i$ of each polynomial function~$f_i(\bm x)$
\begin{align*}
f_1(x_1,x_2,\ldots,x_n)&=0\\
f_2(x_1,x_2,\ldots,x_n)&=0\\
&\vdots\\
f_n(x_1,x_2,\ldots,x_n)&=0
\end{align*}
has an infinite number of zeros, the number of its isolated zeros in~$\C^n$, counting multiplicities, does not exceed the number~$d=d_1 d_2\cdots d_n$. 
\end{mytheo}

Due to the fact that the real numbers are a subset of the complex numbers, the resulting number of zeros in~$\R^n$ is less or equal than the number given by B\'ezout's theorem. For incomplete polynomials Bernstein's theorem allows to calculate a tighter bound for the number of zeros. Since the generated polynomial functions by (\ref{for:poly}) are complete, Bernstein's theorem does not provide a closer boundary.
\end{proof}
%%%%%%%%%%%%%%%%%%%%%%%%%%%%%%%%%%%%%%%%%%%%%%%%%%
%%%%%%%%%%%%%%%%%%%%%%%%%%%%%%%%%%%%%%%%%%%%%%%%%%
\section{Stability}
In this section we analyze the stability of the calculated equilibrium points of deterministic GP-SSMs. For each presented covariance function the related stability condition can be found in the following listing.
%%%%%%%%%%%%%%%%%%%%%%%%%%%%%%%%%%
%Squared exponential
%%%%%%%%%%%%%%%%%%%%%%%%%%%%%%%%%%
\begin{mytheo}[Stability of GP-SSMs with squared exponential covariance function]
A deterministic GP-SSM with squared exponential covariance function and $m$ training points has the following properties:
\begin{itemize}
\item[(i)] There exists an invariant set
\begin{align*}
\Lambda=\left\lbrace \bm x\in\R^n \mid \vert x_i \vert \leq \sigma_{i,f}^2 \sqrt{m} \Vert \bm h(i)\Vert,\forall i=1,\ldots,n  \right\rbrace
\end{align*}
which is also globally attractive.
\item[(ii)] The solution is globally uniformly ultimately bounded with bound $b=\sqrt{m}\left\Vert \left[  \sigma_{1,f}^2 \Vert \bm h(1) \Vert,\ldots, \sigma_{n,f}^2 \Vert \bm h(n) \Vert \right] \right\Vert$.
\end{itemize}
\label{theo:stab_exp}
\end{mytheo}
%We start with some definitions:
%\begin{mydef}
%The nonempty set $\Lambda\subset\R^n$ is invariant for the autonomous system $\bm x_{k+1} = \bm f(\bm x_k)$ if and only if $\forall \bm x_0 \in \Lambda$, the system evolution satisfies $\bm x_k\in\Lambda,\forall k \in\N^*$. 
%\end{mydef}
%An invariant set is a subset in space $\X$ which will never leave once a trajectory of the system enters. Important classes of invariant sets are for example equilibrium points and limit cycles.
%\begin{mydef}
%The nonempty set $\Gamma\subset\R^n$ is called attractive for the autonomous system $\bm x_{k+1} = \bm f(\bm x_k)$ if there is a neighbourhood $N$ of $\Gamma$ and for all $\bm x_0 \in N$ the trajectory $\bm x_k\to\Gamma$ with $k\to\infty$. If the neighbourhood $N$ equals $\R^n$ than the set is globally attractive.
%\end{mydef}
%The uniformly ultimately boundedness of a system is specified by following definition.
%\begin{mydef}
%The solutions of a system \ref{for:discsys} are
%\begin{itemize}
%\item uniformly ultimately bounded with ultimate bound $b$ if there exist positive constants $b$ and $c$ and for every $a\in (0,c)$, there is $T=T(a,b)\geq 0$, such that 
%\begin{align}
%\Vert \bm{x}_0 \Vert \leq a \Rightarrow \Vert \bm{x}_k \Vert \leq b,\,\forall t \geq T. \label{for:defbound}
%\end{align}
%\item globally uniformly ultimately bounded if \ref{for:defbound} holds for arbitrarily large $a$.
%\end{itemize}
%\end{mydef}
\begin{proof}
The proof starts with presenting some properties of the smooth covariance function~$k_{\varphi_i}(\bm x,\bm x^\prime)$. For all~$\sigma_f\in\R_+$ and~$\lambda\in\R^*$ the function is bounded with
\begin{align}
\sup_{\bm x,\bm x^\prime\in\R^n}\hspace{-0.2cm} k_{\varphi_i}(\bm x,\bm x^\prime)&=\left. \sigma_{i,f}^2 \exp{\left(-\frac{\Vert \bm x-\bm x^\prime \Vert^2}{2\lambda^2} \right)}\right|_{\bm x=\bm x^\prime}\hspace{-0.2cm}=\sigma_{i,f}^2\\
\inf_{\bm x,\bm x^\prime\in\R^n}\hspace{-0.2cm}k_{\varphi_i}(\bm x,\bm x^\prime)&=\hspace{-0.3cm}\lim_{\Vert \bm x-\bm x^\prime\Vert\to\infty}\hspace{-0.3cm}\sigma_{i,f}^2 \exp{\left(-\frac{\Vert \bm x-\bm x^\prime \Vert^2}{2\lambda^2} \right)}=0.
\end{align}
According the Cauchy-Schwarz inequality and the results above the following holds:
\begin{align}
\vert\bar x_{i,k+1}\vert&=\vert\bm k_{\varphi_i}(\bm x_k,X)^\top \bm h(i)\vert\leq \sigma_{i,f}^2 \sqrt{m}\Vert \bm h(i)\Vert\label{for:stab:sq}
\end{align}
Therefore, the invariant set~$\Lambda$ is a neighbourhood of zero which is determined by
\begin{align}
\Lambda=\left\lbrace \bm x\in\R^n \mid \vert x_i \vert \leq \sigma_{i,f}^2 \sqrt{m}\Vert \bm h(i)\Vert,\forall i=1,\ldots,n  \right\rbrace
\end{align}
%\vspace{-0.1cm}  
Furthermore, we want to show that the set~$\Lambda$ is attractive, (i). Since (\ref{for:stab:sq}) shows that for any~$\bm x_0\in\X$ the absolute value of the next step state vector~$\bm x_{1}$ is equal or less~$\sigma_{i,f}^2\sqrt{m}\Vert \bm h(i)\Vert$, the state~$\bm x_k$ approaches~$\Lambda$ for~$k\geq 1$. This guarantees globally uniformly ultimately boundedness, (ii), with ultimate bound
%\vspace{-0.5cm}  
\begin{align}
b=\sqrt{m}\left\Vert \left[  \sigma_{1,f}^2 \Vert \bm h(1) \Vert,\ldots, \sigma_{n,f}^2 \Vert \bm h(n) \Vert \right] \right\Vert
\end{align}
\label{proof:invariant_set} 
\end{proof}
An important consequence of Theorem~\ref{theo:stab_exp} is that it is not possible to learn unbounded system trajectories with the GP-SSM which are based on the squared exponential covariance function.
%%%%%%%%%%%%%%%%%%%%%%%%%%%%%%%%%%
%Linear
%%%%%%%%%%%%%%%%%%%%%%%%%%%%%%%%%%
\begin{mytheo}[Stability of GP-SSMs with linear covariance function]
A deterministic GP-SSM with linear covariance function is stable if the spectrum of the state matrix
\begin{align*}
A=\begin{pmatrix}
X_{1,1\dots m}\bm h(1),\ldots,X_{n,1\dots m}\bm h(1)\\
X_{1,1\dots m}\bm h(2),\ldots,X_{n,1\dots m}\bm h(2)\\
\vdots\\
X_{1,1\dots m}\bm h(n),\ldots,X_{n,1\dots m}\bm h(n)\\
\end{pmatrix}
\end{align*}
is equal or less one. If the magnitude is strictly less then one, i.e. ~$\vert \sigma(A)\vert<1$, than the equilibrium point is asymptotically stable. Otherwise, the system is unstable.
\end{mytheo}
\begin{proof}
Since the system dynamic of a GP-SSM with linear covariance function is a linear function, see~(\ref{for:stab:lin}), the theorem about linear stability can be directly applied. 
\end{proof}
%%%%%%%%%%%%%%%%%%%%%%%%%%%%%%%%%%
%Polynomial
%%%%%%%%%%%%%%%%%%%%%%%%%%%%%%%%%%
\begin{mytheo}[Stability of GP-SSMs with polynomial covariance function]
A deterministic GP-SSM with polynomial covariance function is (locally) stable in~$\bm{x^*}$ if the spectrum of the matrix
 \begin{align}
A=\frac{\partial}{\partial \bm{x_k}} \hspace{-1.1cm}\left.\sum_{\hphantom{aaaaaaa} l_1+\ldots+l_{n+1}=p_i}\hspace{-1.1cm}{\alpha_{i,l_1,\ldots,l_{n+1}} x_{1,k}^{l_1} x_{2,k}^{l_2} \cdots x_{n,k}^{l_n}\sigma_{i,0}^{2l_{n+1}}}\right|_{\bm{x_k}=\bm{x^*}}
\end{align}
is equal or less one. If~$\vert \sigma(A)\vert<1$ the equilibrium point is asymptotically stable.
\end{mytheo}
\begin{proof}
The theorem is a direct application of Lyapunovs direct method. Since the polynomial function is smooth, the derivative exists.
\end{proof}
%%%%%%%%%%%%%%%%%%%%%%%%%%%%%%%%%%%%%%%
%%%%%%%%%%%%%%%%%%%%%%%%%%%%%%%%%%%%%%%
\section{Simulations}
\subsection{Equilibrium points}
In this section we want to present some illustrations for the equilibrium sets with different covariance functions. For this purpose, 100 randomly linear systems are generated:
\begin{align}
\bm{x_{k+1}}=\begin{bmatrix}
a_{11} & a_{12}\\a_{21} & a_{22}
\end{bmatrix}
\bm{x_{k}}+n
\label{for:ex_lin_sys}
\end{align}
where~$a_{11},a_{12},a_{21},a_{22}$ are random numbers drawn from the uniform distribution on the open interval~$]0,1[$ and with~$n$ as Gaussian distributed noise~$\mathcal{N}(0,0.05^2)$. Each system is learned by a GP-SSM with 100 homogeneously distributed training points on~$[-1,1]\times[-1,1]$ and 5 different covariance functions (linear, polynomial with~$p=2,3,5$, squared exponential), see Table \ref{tab:kernel}.\\
The hyperparameters are optimized according to the log. likelihood function with a conjugate gradient method. The equilibrium points are numerically estimated by local solvers which start from multiple points in~$[-20,20]\times[-20,20]$ . As comparison, the same procedure is applied with random generated nonlinear system which have a multiple number of equilibrium points:
\begin{align}
\bm{x_{k+1}}=\begin{bmatrix}
\sin(\alpha_1 x_{2,k})+x_{1,k}\\
\sin(\alpha_2 x_{1,k})+x_{2,k}
\end{bmatrix}
+n
\end{align}
where~$\alpha_1,\alpha_2$ are random numbers drawn from the uniform distribution on the open interval~$]0,\frac{3}{2}\pi[$ and with~$n$ representing Gaussian distributed noise~$\mathcal{N}(0,0.05^2)$.
Now, the starting area of the multiple local solvers is~$[-5,5]\times[-5,5]$.
Table~\ref{tab:nr_eq} and Table~\ref{tab:nr_eq_nonlin} show the number of found equilibrium points of the trained GP-SSMs by the linear and the nonlinear systems.
\begin{table}[h]
  \begin{center}
    \begin{filecontents*}{data/nrs_eq.dat}
   0.0000000e+00   1.0000000e+02   0.0000000e+00   0.0000000e+00   0.0000000e+00
   0.0000000e+00   5.3000000e+01   4.4000000e+01   3.0000000e+00   0.0000000e+00
   0.0000000e+00   5.4000000e+01   4.2000000e+01   3.0000000e+00   1.0000000e+00
   0.0000000e+00   5.3000000e+01   4.2000000e+01   4.0000000e+00   1.0000000e+00
   0.0000000e+00   5.0000000e+01   1.0000000e+00   4.9000000e+01   0.0000000e+00
\end{filecontents*}
\pgfplotstabletypeset[
    every head row/.style={after row=\hline},
       create on use/newcol/.style={
        create col/set list={Linear\vphantom{\large A},Polynomial~$p=2$,Polynomial~$p=3$,Polynomial~$p=5$,Squared Exp.}
    },
    columns/newcol/.style={string type,column name=Cov.func./ \# of Equlib.},
    columns={newcol,0,1,2,3,4},
    columns/0/.style={column type=|c},
        ]{data/nrs_eq.dat} % filename/path to file
  \end{center}
  \caption{Number of equilibrium points of 100 GP-SSMs each trained by a randomly generated 2-dimensional, linear systems.\label{tab:nr_eq}}
\end{table}
\begin{table}[h]
\vspace{-0.6cm}
  \begin{center}
    \begin{filecontents*}{data/nrs_eq_nonlin.dat}
   0.0000000e+00   1.0000000e+02   0.0000000e+00   0.0000000e+00   0.0000000e+00   0.0000000e+00
   0.0000000e+00   9.7000000e+01   3.0000000e+00   0.0000000e+00   0.0000000e+00   0.0000000e+00
   0.0000000e+00   7.0000000e+01   0.0000000e+00   5.0000000e+00   2.5000000e+01   0.0000000e+00
   0.0000000e+00   2.7000000e+01   2.0000000e+00   1.0000000e+01   2.7000000e+01   3.4000000e+01
   0.0000000e+00   3.0000000e+00   2.0000000e+00   3.0000000e+01   3.2000000e+01   3.3000000e+01
\end{filecontents*}
\pgfplotstabletypeset[
    every head row/.style={after row=\hline},
       create on use/newcol/.style={
        create col/set list={Linear\vphantom{\large A},Polynomial~$p=2$,Polynomial~$p=3$,Polynomial~$p=5$,Squared Exp.}
    },
    columns/newcol/.style={string type,column name=Cov.func./ \# of Equlib.},
    columns={newcol,0,1,2,3,4,5},
    columns/3/.style={column name={$[3,4]$}},
    columns/4/.style={column name={$[5,9]$}},
    columns/5/.style={column name={$[10,19]$}},
    columns/0/.style={column type=|c},
        ]{data/nrs_eq_nonlin.dat} % filename/path to file
  \end{center}
  \caption{Number of equilibrium points of 100 GP-SSMs trained by randomly generated 2-dimensional, sinusoidal systems.\label{tab:nr_eq_nonlin}}
\end{table}
\vspace{-0.4cm}
Since it is very unlikely that a GP-SSM with linear covariance function trained by the system (\ref{for:ex_lin_sys}) has zero or infinity equilibrium points, the tables shows just an quantity of one. The polynomial covariance function has always equal or less than~$p^2$ equilibrium points and the squared exponential GP-SSMs more than zero.
\subsection{Stability example}
Due to the fact that the squared exponential function is one of the most used covariance function, we present here an example for the boundedness of this kind of GP-SSM. This example deals with the well-known, nonlinear system the Van der Pol oscillator. The discretization of the oscillator is described by~\cite{van2010new} with 
\begin{align}
	x_{k+1}&=\phi(T,x_k,y_k,\epsilon) \Psi(x_k,y_k)T\notag\\
	&+(\varphi(T,x_k,y_k,\epsilon)+1)x_k+n_1\notag\\
	y_{k+1}&=\phi(T,x_k,y_k,\epsilon) \Lambda(x_k,y_k)T\notag\\
	&+(\varphi(T,x_k,y_k,\epsilon)+1)y_k+n_2\label{for:vanderpol}
\end{align} 
where~$T\in\R$ is the sample time and the parameter~$\epsilon\in\R$ a scalar which influence the nonlinearity of the system. For this example~$\epsilon$ is set to~$-0.8$ and the sample time~$T=0.1$. Gaussian distributed noise is added by~$n_1,n_2\sim\mathcal{N}(0,0.01^2)$\\
A GP-SSM with squared exp. covariance function is trained with 441 homogeneous distributed points in~$[-4,4]\times[-4,4]$. The hyperparameters are optimized by the minimization of the log. likelihood function with a conjugate gradient method. Figure~\ref{fig:fig4} shows for~$x_0=-1.8,y_0=0$ the trajectory of the system (\ref{for:vanderpol})~$x_k,y_k$ and the prediction of the trained GP-SSM~$\bar{x}_k,\bar{y}_k$. Since the trajectory stays inside the training area, the predicted trajectory is very similar. Furthermore, the boundedness of the trained GP-SSM is demonstrated.\\
An other example with a different initial value is presented in Fig.~\ref{fig:fig5}. The graph shows the resulting trajectory for the initial values~$x_0=2.2,y_0=0$. 
   \begin{figure}[ht]
   \vspace{-0.2cm}
	\begin{tikzpicture}
\begin{axis}[
  xlabel={Time [s]},
  ylabel=State,
  line width=0.7pt,
  grid style={dashed,gray},
  grid = both,
  height=6.2cm,
  width=\columnwidth]
\addplot+[mark=+] table [x index=0,y index=1]{data/gt1.dat};
\addlegendentry{$x_k$}
\addplot+[mark=+] table [x index=2,y index=3]{data/gt1.dat};
\addlegendentry{$y_k$}
\addplot+[mark=x,color=green] table [x index=0,y index=1]{data/gp1.dat};
\addlegendentry{$\bar{x}_k$}
\addplot+[mark=x] table [x index=2,y index=3]{data/gp1.dat};
\addlegendentry{$\bar{y}_k$}
\end{axis}
\end{tikzpicture} 
      \normalfont{\vspace{-0.2cm}\caption{The prediction~$\bar{x}_k,\bar{y}_k$ of a GP-SSM with squared exponential covariance function is always bounded. With~$x_0=-1.8,y_0=0$ the prediction and the trajectory of (\ref{for:vanderpol}) are quite similar.\label{fig:fig4}}} 
   \end{figure}
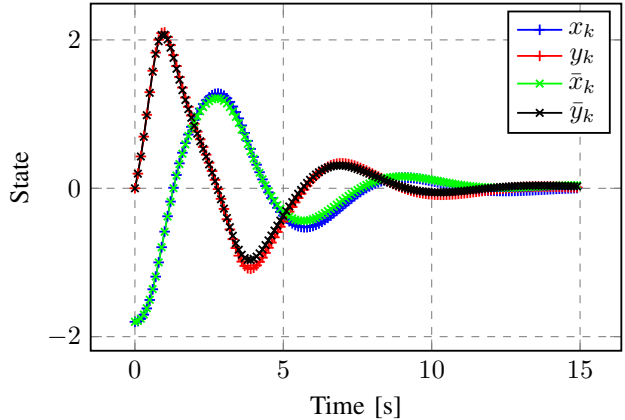
\begin{figure}[ht]
\vspace{-0.5cm}
\begin{tikzpicture}
\begin{axis}[
  xlabel={Time [s]},
  ylabel=State,
  legend pos=south west,
  grid style={dashed,gray},
  grid = both,
  ymin=-35,ymax=25,
  line width=0.7pt,
  height=6.2cm,
  width=\columnwidth]
      \addplot+[mark=+] table [x index=0,y index=1]{data/gt2.dat};
\addlegendentry{$x_k$}
\addplot+[mark=+] table [x index=2,y index=3]{data/gt2.dat};
\addlegendentry{$y_k$}
\addplot+[mark=x,color=green] table [x index=0,y index=1]{data/gp2.dat};
\addlegendentry{$\bar{x}_k$}
\addplot+[mark=x] table [x index=2,y index=3]{data/gp2.dat};
\addlegendentry{$\bar{y}_k$}
\end{axis}
\end{tikzpicture} 
      \normalfont{\vspace{-0.2cm}\caption{The prediction~$\bar{x}_k,\bar{y}_k$ of a GP-SSM with squared exponential covariance function is always bounded even if the trajectory~$x_k,y_k$ of the original system is unbounded.\label{fig:fig5}}}
      \vspace{-0.2cm}
\end{figure}
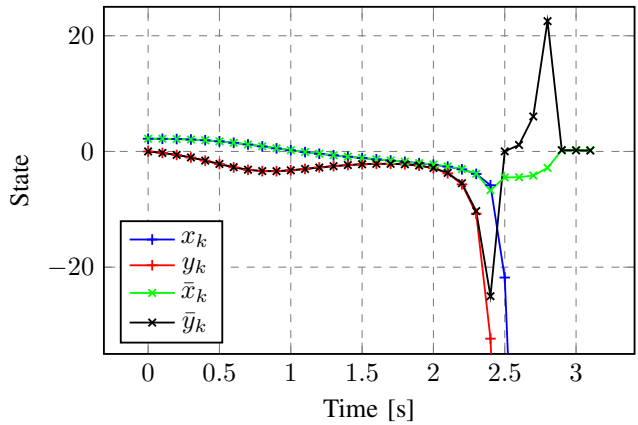
%%%%%%%%%%%%%%%%%%%%%%%%%%%%%%%%%%%%%%%%%%%%%%%%%%%%%%%%%%%%%%%%%%%%%%%%%%%%%%%%
%\addtolength{\textheight}{-15.5cm}   % This command serves to balance the column lengths
This initial point is not inside the attraction area of the oscillator and the trajectory~$x_k,y_k$ of the system is not bounded. Although the original trajectory is unstable, the prediction of the GP-SSM is bounded.
\section*{Conclusion}
In this paper, we investigate the equilibrium points and stability properties of Gaussian Process State Space Models (GP-SSMs) with different covariance functions. In particular, we study GP-SSMs with squared exponential, linear, and polynomial covariance function. A deterministic GP-SSM with the widely spread squared exponential covariance function generates always at least one equilibrium and is globally uniformly ultimately bounded. Therefore, it is not possible to learn unbounded trajectories with this approach.\\
The linear covariance function generates one equilibrium point except for pathological cases. The number of equilibrium points of a GP-SSM with polynomial function is always equal or less than the degree of the polynomial. Two examples visualize the shown properties.

\section*{ACKNOWLEDGMENTS}
The research leading to these results has received funding from the European Research Council under the European Union Seventh Framework Program (FP7/2007-2013) / ERC Starting Grant ``Control based on Human Models (con-humo)'' agreement n\textsuperscript{o}337654.

%%%%%%%%%%%%%%%%%%%%%%%%%%%%%%%%%%%%%%%%%%%%%%%%%%%%%%%%%%%%%%%%%%%%%%%%%%%%%%%%

\bibliography{mybib}
\bibliographystyle{ieeetr}

\end{document}